\documentclass[]{article}
\usepackage{amsfonts}
\usepackage{amsthm}
\usepackage{youngtab}
\usepackage{amsmath,amscd}
\usepackage{graphicx}
\usepackage{caption}
\usepackage{subcaption}
\usepackage{listings} 
\usepackage{amssymb}
\usepackage{algpseudocode}
\usepackage{mathtools}
\usepackage[linesnumbered,vlined,ruled,commentsnumbered]{algorithm2e}

\makeatletter
\def\BState{\State\hskip-\ALG@thistlm}
\makeatother

\newtheorem{theorem}{Theorem}[section]

\theoremstyle{definition}
\newtheorem{defn}{Definition}[section]
\newtheorem{lem}{Lemma}[section]

\newtheorem{exmp}{Example}[section]

\title{Majorisation-minimisation algorithms for minimising the difference between lattice submodular functions}
\author{
	Conor McMeel\\
	\texttt{c.mcmeel18@imperial.ac.uk}
	\and
	Panos Parpas\\
	\texttt{panos.parpas@imperial.ac.uk}
}

\begin{document}
	\maketitle
\begin{abstract}
	We consider the problem of minimising functions represented as a difference of lattice submodular functions. We propose analogues to the SupSub, SubSup and ModMod routines for lattice submodular functions. We show that our majorisation-minimisation algorithms produce iterates that monotonically decrease, and that we converge to a local minimum. We also extend additive hardness results, and show that a broad range of functions can be expressed as the difference of submodular functions.
\end{abstract}

\section{Introduction}
In discrete optimisation, many objectives are expressible as submodular minimisation or maximisation problems. This is useful, as submodular functions can be minimised in strongly polynomial time \cite{iwata2001combinatorial}, and have strong constant approximation factors under a variety of constraints for maximisation \cite{lee2009non} \cite{buchbinder2014submodular}. Applications for these problems have included medical imaging \cite{hoi2006batch}, document summarisation \cite{lin2012learning} and online advertising \cite{alaei2010maximizing}.

We refer to $V = \{1, \ldots, n\}$ as the ground set, and a function $f: 2^V \rightarrow \mathbb{R}$ is then called submodular if the following inequality holds for all $A, B \subseteq V$:
\begin{equation}
	f(A) + f(B) \geq f(A \cup B) + f(A \cap B).
\end{equation}
We note that a subset $A$ can also be represented as a $0-1$ vector $x$, where a component takes on the value $1$ if the element is present in the subset $A$. In this way our inequality can be written as
\begin{equation}
f(x) + f(y) \geq f(\min(x,y)) + f(\max(x,y))
\end{equation}
for all $x, y \in \{0,1\}^n$. A popular and often-studied application of submodular optimisation which we shall use as our running example is that of sensor placement. Suppose we wish to detect the level of air pollution in a town. We have a finite set of locations that we can install air pollution sensors in, but we are constrained by some budget. Knowing the geography of our town, we wish to select the set of locations such that the information we receive is maximised. As shown in \cite{radanovic2015incentive}, this problem is submodular.

In the classic $0-1$ case, this corresponds to only having one type of sensor. However, we can suppose that we have a number of different types of sensors available, each with a strength level that can be parameterised by an integer variable $x \in \{0, \ldots, k\}$. Lattice submodularity is then determined by satisfying Equation $(2)$, but for all $x, y \in \{0, \ldots, k\}^V$. Note that in general, we can allow $k$ to be different for different elements.

Submodular functions in the set case also lend themselves to a useful diminishing returns property, which is often considered as an equivalent definition.
\begin{defn}[Diminishing returns property]
	A submodular set function can also be defined by the following. For any $A \subset B \subset V$ and $e \notin A$, we have
	\begin{equation}
		f(A \cup e) - f(A) \geq f(B \cup e) - f(B).
	\end{equation}
\end{defn}
However, for lattice functions we unfortunately do not get that for free \cite{soma2018maximizing}, so we define a proper subclass known as DR-submodular functions as follows:
\begin{defn}[DR-Submodular]
	Let $f$ be a lattice submodular function on $\prod_i \{0, \ldots, k\}^V$. Then $f$ is called a DR-submodular function if for all $x, y \in \prod_i \{0, \ldots, k\}^V$ with $x \leq y$ componentwise, $c > 0$ and basis vector $e_j$, we have
	\begin{equation}
		f(x + ce_j) - f(x) \geq f(y + ce_j) - f(y).
	\end{equation}
\end{defn}
As shown in \cite{bach2016submodular}, many typical objectives we come across are in fact DR-submodular, so it makes sense to consider this subclass. The specific problem of optimising DR-submodular functions has been studied before \cite{bian2017continuous} \cite{soma2018maximizing}. Across these papers, the problem of maximisation has been studied in monotone and non-monotone settings across a variety of constraints. 

In this paper, we will carry out our majorisation-minimisation algorithm by discretising the lattice submodular functions, giving us submodular functions on some lattice. In \cite{ene2016reduction}, the authors showed that DR-submodular functions on a lattice can be reduced to the set function submodular case, where the ground set will be smaller. Due to this, we do not consider the specific DR subclass here, as it would be more efficient to carry out this reduction and then use the work of \cite{iyer2012algorithms}.

To motivate the problem of considering the difference of submodular functions, we turn back to the sensor information problem. Here, we may wish to also factor in the costs of sensors also. As is typical in real-world applications, we expect to get a bulk discount if we buy many sensors. Additionally, we may expect that for any one sensor, as the sensor gets stronger the unit price increases slower. This leads us to the diminishing returns property for submodularity. Specifically, if we obtain information $f(x)$, and spend $c(x)$ to get that information, for some sensor placement $x$, we wish to minimise $f(x) - \lambda c(x)$, for some tradeoff parameter $\lambda$. This formulation also applies to any other problem of optimising a submodular function with some cost.

In the set function case, this problem was considered by \cite{iyer2012algorithms}, who proposed several different majorisation minimisation techniques. Here, we will use the work of \cite{bach2016submodular} to come up with our algorithm, who worked on discretised versions of continuous submodular function, hence lattice submodular functions. We will derive majorisation-minimisation algorithms on the difference $v(x) = f(x) - g(x)$. It can be considered the extension to the lattice case of the SupSub procedure in \cite{iyer2012algorithms}, originally introduced in \cite{narasimhan2012submodular}, as well as the SubSup and ModMod procedures. In fact by discretising continuous functions, we can also consider it an extension to the continuous domain much like \cite{bach2016submodular}. We will show that we converge to a local minimum of the function, and consider the class of functions we can represent as the difference of submodular functions.

\subsection{Outline}
We begin by discussing some preliminaries Section $2$. We consider the theory of continuous and lattice submodular functions, and show how subgradients are obtained. We show modular and tight lower and upper bounds, along with a decomposition that allows us to utilise the theory of DR-submodular functions effectively. Section $3$ presents and discusses the three algorithms, along with a discussion of a stopping criterion. In Section $4$, we discuss complexity, deriving the functions $f, g$ for the representation $v(x) = f(x)-g(x)$, along with other theoretical issues. We conclude with a brief discussion in Section $5$.

\section{Preliminaries}
We are considering the problem of minimising the difference between lattice submodular functions. In particular, we seek to minimise
\begin{equation}
v(x) = f(x) - g(x)
\end{equation}
where $f, g$ are lattice submodular. For our algorithms, our plan will be to in fact use an extension of the functions $f$ that was defined in \cite{bach2016submodular}. While doing this, we will also give the main ingredient for computing the lower bound.

\subsection{An Extension}

Our extension will be from the set $\prod_i \{0, \ldots, k_i -1\} = \prod_i \mathcal{X}_i$ to the product of measures of $\mathcal{X}_i$. The motivation for this can be thought of as follows. In the Lovasz extension, the value $x \in [0,1]$ taken on by a component can be thought of as the probability that an element $x$ is included in a set. Similarly, we will extend here to probability mass functions over the domain $\{0, \ldots , k_i-1\}$.

We note that any PMF on this domain can be represented by its reverse cumulative probability density function $\rho_i(x_i) = \mu_i(x_i) + \ldots + \mu_i(k_i-1) = F_{\mu_i}(x_i)$. We then see that $\rho_i(0) = 1$ always, and the only constraint on it is that its elements are non-decreasing, and all belong to $[0,1]$. As in \cite{bach2016submodular}, we denote this set of vectors as $[0,1]^{k_i-1}_{\downarrow}$.

Next, we think of this set of vectors just defined, and note that if we mark the locations of $\rho_i(x_i)$, we will divide the number line into $k_i$ segments. This induces a map $\theta(\rho_i, t)$ from $\mathbb{R}$ to ${0,\ldots, k_i-1}$ with $\theta(\rho_i,t) = k_i-1$ if $t \leq \rho_i(k_i-1)$, $\theta(\rho_i,t) = x_i$ if $t \in (\rho_i(x_{i+1}, \rho_i(x_i)]$ for $x_i \in {1, \ldots, k_i-2}$, and $\theta(\rho_i,t) = 0$ for $t \geq \rho_i(1)$.

Our extension is then thus defined, for a submodular function $f$

\begin{equation}
f_\downarrow(\rho) = \int_0^1 f(\theta(\rho_1,t), \ldots, \theta(\rho_n,t)) dt.
\end{equation}
For a more thorough exposition, see \cite{bach2016submodular}. We now detail how the extension is evaluated using a greedy algorithm.

\subsection{A lower bound}

We now give the greedy algorithm from \cite{bach2016submodular}.

\begin{theorem}
	Consider the extension of $f_\downarrow(\rho)$ of a submodular function $f$. Order all values of $\rho$ in decreasing order, breaking ties arbitrarily but ensuring that all components for a given $\rho_i$ are in the correct order. Assume the value at position $s$ is equal to $t(s)$ and corresponds to $\rho_{i(s)}(j(s))$. Then we have the following expression for $f_\downarrow$
	
	\begin{equation}
		f_\downarrow(\rho) = f(0) + \sum_{i=1}^r t(s)(f(y(s)) - f(y(s-1)))
	\end{equation}
	which we write in the form $f_\downarrow(\rho) = f(0) + \sum_{i=1}^n \sum_{x_i=1}^{k_i-1} w_i(x_i)\rho_i(x_i)$, where $w_i(x_i)$ corresponds to the difference in value between two function evaluations whose arguments differ by a single basis vector.
\end{theorem}

In this section, we will first remind ourselves of the different characterisations of the base polyhedron, how the greedy algorithm is used to construct extreme points for one of these characterisations, and how subgradients are computed. Using this, we will compute a lower bound that will be used for our algorithm.

We remind ourselves first of the characterisation of the base polyhedron as given initially in \cite{bach2016submodular}:
\begin{defn}[Base Polyhedron]
	Let $f(x_1, \ldots, x_n)$ be a discrete submodular function WLOG with $f(0,0,\ldots, 0) = 0$, in each argument from $0$ to the integer $k_i-1$ respectively. Then the submodular polyhedron can be defined by arguments $w_i \in \prod_i \mathbb{R}^{k_i-1}$ such that for all $(x_1, \ldots, x_n)$:
	\begin{equation}
	\sum_{i=1}^n \sum_{y_i = 1}^{x_i} w_i(y_i) \leq f(x_1, \ldots, x_n),
	\end{equation}
	\begin{equation}
	\sum_{i=1}^n \sum_{y_i = 1}^{k_i-1} w_i(y_i) = f(k_1-1, \ldots, k_n-1).
	\end{equation}
\end{defn}
As mentioned in \cite{bach2016submodular}, this polyhedron is in fact not a polyhedron, and is unbounded, if there is any $k_i > 2$. This can be made explicit in the following example.
\begin{exmp}[Base Polyhedron Unbounded]
	Let $k_i = 3$ for some $i$. Then for that $i$, we can add to $w_i$ any $u_i$ such that
	\begin{equation}
	u_1(1) = -1, u_1(2) = 1,
	\end{equation}
	and we see this won't violate any of the equations in the definition. This argument extends straightforwardly to any $k_i > 2$.
\end{exmp}
Because of this, they instead define the base polyhedron as the convex hull of outputs of a greedy algorithm. The following result shows that the base polyhedron still behaves in the same way:

\begin{lem}
	Let $f$ be some submodular function. Then for any $\rho \in \prod_i \mathbb{R}_\downarrow^{k_i-1}$ we have that
	\begin{equation}
		\max_w \langle w, \rho \rangle = f_\downarrow(\rho) - f(0),
	\end{equation}
	where we take the max over either characterisation of the base polyhedron.
\end{lem}
We see that as long as we take a $w$ compatible with the ordering $\rho$, we will get a subgradient for the function $f_\downarrow(\rho)$. Restricting to the $\rho$ that will give us points in the domain of our original submodular function $f$, we get an element of the subdifferential of $f$.

Now we can construct a lower bound. First, we construct a set $A$ with $r$ elements, each of which corresponds to an increment of one of the $n$ basis vectors. There are $k_i - 1$ copies of the increment of element $i$. Note that such a set $A$ corresponds to a chain defined by
\begin{equation}
0 = p_0 < p_1 < \ldots < p_r = (k_1-1, \ldots, k_n-1). \label{chain}
\end{equation}
Take a permutation of $A$ denoted by $\sigma$ and form its corresponding chain $\mathcal{C}_\sigma$. Ensure $\sigma$ is such that the chain $p_i$ \textit{contains} $z$, something we define now.
\begin{defn}[Chain containing an element]
	Let $p_i$ be a chain as defined in Equation \eqref{chain}. We say the chain \textit{contains} $x$, for some vector $x = (x_1, \ldots, x_n)$, if we have $ p_{x_1 + \ldots + x_n} = x$.
\end{defn}
Note that any chain that contains a vector $y$ is compatible with the ordering of the $\rho$ that corresponds to $y$. In particular, we'll take the chain $C_{\sigma}$ where we increment the first element to $y_1$, then the second to $y_2$ and so on. After it reaches $y$, we let it have any behaviour. Then we can form $w$ by taking differences of successive elements of the chain.

Now we can form the function corresponding to the lower bound by making $w$ as described earlier. This will be denoted by
\begin{equation}
h_{f,y}(i,j) = w_i(j).
\end{equation}
To extend this definition to an entire point $x = (x_1, x_2 ,\ldots, x_n)$ we can do the following:
\begin{equation}
h_{f,y}(x) = \sum_{i = 1}^n \sum_{j = 1}^{x_i} h_{f,y}(i, j) = \sum_{i = 1}^n \sum_{j = 1}^{x_i} w_i(j).
\end{equation}
Now note that for each $p_i$ in the chain, we have that
\begin{equation}
h_{f,y}(p_i) = f(p_i).
\end{equation}
In particular, note that as $y$ is contained in this chain, we have that
\begin{equation}
h_{f,y}(y) = f(y).
\end{equation}
Then due to the fact that this is a subgradient, we have that
\begin{equation}
h_{f,y}(x) \leq f(x)
\end{equation}
for all $x$. This is parameterised by $y, \sigma$ and is tight at $y$.

\subsection{Upper Bound}
In \cite{iyer2012algorithms}, an upper bound was derived for a submodular set function. This can be extended, but it does require DR-submodularity. We show now how to get around that as follows if we know one particular quantity:

\begin{lem}
	Let $f$ be a lattice submodular function. Then $f$ can be represented as the sum of a modular function $g$ and a DR-submodular function $h$.
\end{lem}
\begin{proof}
	Let $\lambda$ be the largest second difference of the function $f$ that violates the DR property, or at least an upper bound of it. Because of submodularity, we know the difference will be a second difference entirely within one basis element. Namely, take:
	
	\begin{equation}
		\lambda \leq \max_{x, i} \left(f(x + 2e_i) - 2f(x+e_i) + f(x)\right).
	\end{equation}
	
	Then let $g = \lambda(x_1^2 + \ldots + x_n^2)$. This will give us $h = f-g$ DR-submodular as required.
\end{proof}
Computing the exact tight value of $\lambda$ is often hard, but there are some cases where it will be easier to derive upper bounds:
\begin{enumerate}
	\item If the function $f$ is the discretisation of a continuous function, we can compute the largest positive eigenvalue, or an upper bound of it, via a number of eigenvalue algorithms.
	\item If the function $f$ is also $L^\natural$-convex (or midpoint convex), the function is convex-extensible \cite{moriguchi2012discrete}, and so proceed similarly to above.
	\item If the function $f$ is a quadratic program, with $f(x) = x^T A x + b^Tx + c$, then we have that $\lambda$ is the maximum positive diagonal element of $A$ (or $0$ if none exist).
\end{enumerate}

The idea now here is that we will let $f = g+h$, and then derive a modular and tight upper bound for $h$, thus giving our modular and tight upper bound for $f$ as $g$ is already modular. The aim for that is to generalise the bounds found in \cite{iyer2012algorithms}, given as

\begin{align}
f(Y) &\leq f(X) - \sum_{j \in X \setminus Y} f(j | X \setminus j) + \sum_{j \in Y \setminus X} f(j | \emptyset), \\
f(Y) &\leq f(X) - \sum_{j \in X \setminus Y} f(j | V \setminus j) + \sum_{j \in Y \setminus X} f(j | X).
\end{align}
In this section, we assume that $f$ is a lattice submodular function given on the product of sets $\{0, \ldots, k_i-1\}$. The extension of these bounds is given in the following lemma:
\begin{lem}
	Let $f$ be a DR-submodular function as described. Let $m(x) = \max(x,0)$ with $m$ extending its arguments to vectors componentwise. Let $x, y$ be vectors in $\prod_i\{0, \ldots, k_i\}$. Let
	\begin{align}
	(a_1, \ldots, a_n) &= m(x-y), \\
	(b_1, \ldots, b_n) &= m(y-x).
	\end{align} Then we have the following:
	\begin{align}
	f(y) &\leq f(x) - \sum_{i=1}^n [f(x) - f(x - a_ie_i)] + \sum_{i=1}^n f(b_ie_i), \\
	f(y) &\leq f(x) - \sum_{i=1}^n [f(k_{max}) - f(k_{max} - a_ie_i)] + \sum_{i=1}^n f(f(x+b_ie_i) - f(x)),
	\end{align}
	where $e_i$ are the usual basis vectors.
\end{lem}
\begin{proof}
	We first show the following bounds:
	\begin{align}
	f(y) &\leq f(x) - \sum_{i=1}^n [f(x) - f(x - a_ie_i)] + \sum_{i=1}^n [f(b_ie_i+ \min(x,y)) - f(\min(x,y))], \\
	f(y) &\leq f(x) - \sum_{i=1}^n [f(\max(x,y)) - f(\max(x,y) - a_ie_i)] + \sum_{i=1}^n f(f(x+b_ie_i) - f(x)).
	\end{align}
	The proof proceeds similarly to the derivation in \cite{iyer2012algorithms}, but with unions and intersections replaced with mins and maxes respectively. We start with the second statement. Take an arbitrary $x,y$ with $a_i, b_i$ as before. Then note we have
	\begin{align}
	f(\max(x,y)) - f(x) &= \sum_{i=1}^{n} [f(x + \sum_{j=1}^i b_{j}e_{j}) - f(x + \sum_{j=1}^{i-1} b_{j}e_{j})] \\
	&= \sum_{i=1}^n \rho_{b_i, e_i} (x + \sum_{j=1}^{i-1} b_{j}e_{j}),
	\end{align}
	where we take $a_0, e_0 = 0$. Here $\rho_{a,b}(c)$ denotes the marginal return on adding the value $a$ in the basis vector $b$ when the function already has argument $c$. Using the DR-submodular property, we then see
	\begin{equation}
	\sum_{i=1}^n \rho_{b_i, e_i} (x + \sum_{j=1}^{i-1} b_{j}e_{j}) \leq \sum_{i=1}^n \rho_{b_i, e_i} (x) = \sum_{i=1}^n f(x+b_ie_i) - f(x).
	\end{equation}
	Similarly, we'll now consider the following expression:
	\begin{align}
	f(\max(x,y)) - f(y) &= \sum_{i=1}^{n} [f(y + \sum_{j=1}^i a_{j}e_{j}) - f(y + \sum_{j=1}^{i-1} a_{j}e_{j})] \\
	&= \sum_{i=1}^n \rho_{a_i, e_i} (y + \sum_{j=1}^{i} a_{j}e_{j} - a_ie_i) \\
	&\geq \rho_{a_i, e_i} (\max(x,y) - a_ie_i) = \sum_{i=1}^n [f(\max(x,y)) - f(\max(x,y) - a_i)].
	\end{align}
	Subtracting these two gives us the required result. We now proceed to the first statement. We get the first inequality similarly to just how we proceeded:
	\begin{align}
	f(x) - f(\min(x,y)) &= \sum_{i=1}^{n} [f(x + \sum_{j=1}^i a_{j}e_{j}) - f(x + \sum_{j=1}^{i-1} a_{j}e_{j})] \\
	&= \sum_{i=1}^n \rho_{a_i, e_i} (y + \sum_{j=1}^{i} a_{j}e_{j} - a_ie_i) \\
	&\leq \rho_{a_i, e_i} (x - a_ie_i) = \sum_{i=1}^n [f(x) - f(x - a_ie_i)].
	\end{align}
	The second also coming easily:
	\begin{align}
	f(y) - f(\min(x,y)) &= \sum_{i=1}^{n} [f(y + \sum_{j=1}^i b_{j}e_{j}) - f(y + \sum_{j=1}^{i-1} b_{j}e_{j})] \\
	&= \sum_{i=1}^n \rho_{b_i, e_i} (y + \sum_{j=1}^{i-1} b_{j}e_{j}), \\
	&\leq \sum_{i=1}^n \rho_{b_i, e_i} (\min(x,y)) \\
	&= \sum_{i=1 y-x}^n [f(b_ie_i+ \min(x,y)) - f(\min(x,y))].
	\end{align}
	And again we subtract to get the first bounds that we wanted. To get to the required result, simply apply DR-submodularity to the final term of the first bound and the second term of the second bound.
\end{proof}
We note that these bounds are separable. Additionally, we see that these bounds are tight at $x$, namely equality is achieved with $y = x$. This result also applies as written for continuous submodular functions, where the vectors $x, y$ instead belong to $[0,1]^n$ WLOG.

\section{Three algorithms}
We describe the integer lattice majorisation-minimisation algorithms here, in Algorithms $1,2,3$. In Algorithm $1$, note that at every step we are minimising a lattice submodular function, which as shown in \cite{bach2016submodular} to get this to arbitrary precision we have complexity $\mathcal{O}((\frac{2GBn}{\varepsilon})^3\log(\frac{2GBn}{\varepsilon}))$, for a continuous submodular function defined on $[0,B]^n$ with Lipschitz constant $G$ (note that the author minimises a discretised version of the continuous function). For Algorithm $2$, we are instead maximising, for which we have an approximation factor of $1/3$ \cite{gottschalk2015submodular} and runs in $\mathcal{O}(kn)$ calls. For Algorithm $3$, we are at each point minimising a modular function, which can be done easily in $\mathcal{O}(k')$ function evaluations, where we have $k' = \sum_i k_i$, by evaluating each separated function at every point and taking the $n$ minima. Additionally, we note that Algorithms $2,3$ will require in principle an upper bound on the quantity $\lambda$ as described in Lemma $2.2$.

We note that we have quite easily:
\begin{align}
v(x_{t+1}) &= f(x_{t+1}) - h^g_\sigma(x_{t+1}) \\
&\leq m^f_{x_t}(x_t) - h^g_\sigma(x_t) \\
&= f(x_t) - g(x_t) = v(x_t),
\end{align}
the second to last equality coming from the demonstrated tightness of our bounds. However, we note that the sequence may not strictly decrease, and we seek a convergence condition. For the upper bound, we can simply say try both of them, and if neither strictly decreases the function, we're done. For the lower bound, there are many permutations we can change and we want some sort of a stopping criterion.

\IncMargin{1em}
\begin{algorithm}
	\SetKwData{Left}{left}\SetKwData{This}{this}\SetKwData{Up}{up}
	\SetKwFunction{Union}{Union}\SetKwFunction{FindCompress}{FindCompress}
	\SetKwInOut{Input}{input}\SetKwInOut{Output}{output}
	
	\Input{Function $v(x) = f(x) - g(x)$, where $f$ and $g$ are both submodular.}
	\Output{A local minimum of $v(x)$}
	\BlankLine
	Discretise the function $v(x)$ in a preferred way\;
	$x_0 \leftarrow (0,\ldots,0)$; $t \leftarrow 0$\;
	\While{not converged ($x_{t+1} \neq x_t)$}{
		Form the extension $f_\downarrow(x_t)$ of $f(x_t)$.\;
		Choose a permutation $\sigma$ such that the induced chain is compatible with $x_t$\;
		$x_{t+1} = \arg \min_{x} f_{\downarrow}(x) - h^g_\sigma(x)$\;
		t = t+1\;
	}
	
	\caption{Integer lattice SubSup algorithm.}
	\label{algo_disjdecomp}
\end{algorithm}
\DecMargin{1em}

\begin{lem}
	If we choose $\mathcal{O}(n)$ permutations each with different increments directly before and after $x_t$, and attempt to decrease with both upper bounds and we are not successful, then we have reached a local minimum.
\end{lem}
\begin{proof}
	Proof is functionally identical to the set function case as in \cite{iyer2012algorithms}, except instead of saying we consider all $g(X \cup j), g(X \setminus j)$, we just consider all $g(x + e_i), g(x-e_i)$ where $e_i$ are basis vectors.
\end{proof}

\IncMargin{1em}
\begin{algorithm}
	\SetKwData{Left}{left}\SetKwData{This}{this}\SetKwData{Up}{up}
	\SetKwFunction{Union}{Union}\SetKwFunction{FindCompress}{FindCompress}
	\SetKwInOut{Input}{input}\SetKwInOut{Output}{output}
	
	\Input{Function $v(x) = f(x) - g(x)$, where $f$ and $g$ are both submodular and have bounded Hessian.}
	\Output{A local minimum of $v(x)$}
	\BlankLine
	
	Discretise the function $v(x)$ in a preferred way\;
	$x_0 \leftarrow (0,\ldots,0)$; $t \leftarrow 0$\;
	\While{not converged ($x_{t+1} \neq x_t)$}{
		$x_{t+1} = \arg \min_{x} m^f_{x_t}(x) - g(x)$\;
		t = t+1\;
	}
	
	\caption{Integer lattice SupSub algorithm.}
	\label{algo_disjdecomp}
\end{algorithm}
\DecMargin{1em}

\IncMargin{1em}
\begin{algorithm}
	\SetKwData{Left}{left}\SetKwData{This}{this}\SetKwData{Up}{up}
	\SetKwFunction{Union}{Union}\SetKwFunction{FindCompress}{FindCompress}
	\SetKwInOut{Input}{input}\SetKwInOut{Output}{output}
	
	\Input{Function $v(x) = f(x) - g(x)$, where $f$ and $g$ are both submodular and have bounded Hessian.}
	\Output{A local minimum of $v(x)$}
	\BlankLine
	
	Discretise the function $v(x)$ in a preferred way\;
	$x_0 \leftarrow (0,\ldots,0)$; $t \leftarrow 0$\;
	\While{not converged ($x_{t+1} \neq x_t)$}{
		Choose a permutation $\sigma$ such that the induced chain is compatible with $x_t$\;
		$x_{t+1} = \arg \min_{x} m^f_{x_t}(x) - h^g_\sigma(x)$\;
		t = t+1\;
	}
	
	\caption{Integer lattice ModMod algorithm.}
	\label{algo_disjdecomp}
\end{algorithm}
\DecMargin{1em}

\section{Theoretical Analysis}
We note that we don't get any multiplicative approximation guarantees, as the hardness results are inherited from the set function case. However, we would like to extend the additive hardness results of \cite{iyer2012algorithms}. While following their proof will rely on $f$ being DR-submodular, we recall we can write $f = g+h$ for $g$ modular and $h$ DR-submodular via Lemma $2.2$.

So now we act just on the DR-submodular function. This requires an extension of the decomposition from \cite{cunningham1983decomposition} which we give in a slightly weaker form:

\begin{lem}
	Let $f$ be any DR-submodular function with $f(0) = 0$. It can be decomposed into a modular function $g$ plus a monotone function $h$ with $h(0) = 0$.
\end{lem}
\begin{proof}
	We first construct the modular function $g$, then show that the function $h = f - g$ is monotone, and takes the value $0$ at $0$. For any input $y$ we form:
	\begin{equation}
	g(y) = \sum_{k = 1}^n \sum_{j = 1}^{y_k} \frac{m_{j,k}}{j},
	\end{equation}
	for $0 < j \leq y_k$ for each $k$. We note this decays to the modular function in \cite{cunningham1983decomposition} if we restrict all $y_k$ to $\{0, 1\}$. This function is clearly modular. Additionally, it is clear that $h = f-g$ has $h(0) = 0$. To show that it is monotone:
	\begin{align}
	f'(k + e_i) &= f(k + e_i) - \sum_{k = 1}^n \sum_{j = 1}^{y'_k} \frac{m_{j,k}}{j} \\
	f'(k) &= f(k) - \sum_{k = 1}^n \sum_{j = 1}^{y_k} \frac{m_{j,k}}{j}.
	\end{align}
	Note that in the first equation, there will be one extra term subtracted, $\frac{m_{i,k_i + 1}}{k_i+1}$. We claim the right hand side of the first equation here is greater than the second, as:
	\begin{align}
	&f(k + e_i) - \sum_{k = 1}^n \sum_{j = 1}^{y'_k} \frac{m_{j,k}}{j} - (f(k) - \sum_{k = 1}^n \sum_{j = 1}^{y_k} \frac{m_{j,k}}{j}) \\
	&= f(k+e_i) - f(k) - \frac{m_{i,k_i + 1}}{k_i+1} \\
	&= \frac{1}{k_i+1}((k_i+1)(f(k+e_i) - f(k)) - (f(k_{max}) - f(k_{max} - e_i k_i + 1))).
	\end{align}
	We then split up the second term to get:
	\begin{equation}
	f(k_{max}) - f(k_{max} - e_ik_i + 1)) = \sum_{j = 1}^{k_1} f(k_{max} - (j-1)e_i) - f(k_{max} - je_i),
	\end{equation}
	and use the DR property to bound each individual term by $f(k+e_i) - f(k)$, giving us our result.
\end{proof}
Using this, we now have the following combining the two decompositions:
\begin{theorem}
	Let $f$ be a lattice submodular function with $f(0) = 0$. Then $f$ can be represented as the sum of a modular function $g$ and a monotone submodular function $h$, with $h(0) = 0$.
\end{theorem}
This will allow us to get additive bounds similar to \cite{iyer2012algorithms}:
\begin{lem}
	Consider the problem of minimising $f(x) - g(x)$. Apply our decomposition to instead have $v(x) = f'(x) - g'(x) + k(x)$. Here $k$ is modular, and $f, g$ are monotone. Then we have:
	\begin{align}
	\min v(x) &\geq \min_x f'(x) + k(x) - g'(k_{max}), \\
	\min v(x) &\geq f'(0) - g'(k_{max}) + \sum_{k=1}^n (\min_y \sum_{i = 1}^y \frac{m_{i,k}}{i} ).  \\
	\end{align}
\end{lem}
\begin{proof}
	We have:
	\begin{align}
	\min_x v(x) &= \min_x f'(x) - g'(x) + k(x) \\
	&\geq \min_x(f'(x) + k(x)) - \max_x g'(x) \\
	&= \min_x(f'(x) + k(x)) - g'(k_{max}),
	\end{align}
	by monotonicity of $g'$. Next, we continue from the previous line:
	\begin{align}
	&\geq \min_x f'(x) + \min_x k(x) - g'(k_{max}) \\
	&= f'(0) - g'(k_{max}) + \sum_{k=1}^n (\min_y \sum_{i = 1}^y \frac{m_{i,k}}{i} ).
	\end{align}
\end{proof}
The nested optimisation problem here in the second bound is non-trivial, but for our purposes it doesn't actually need to be solved. All we need to see is that by setting $y$ to zero, we can eliminate this term and thus by the monotonicity of $f', g'$, we will always get that this second bound is less than zero.

We also prove some other results that show this algorithm can be broadly applied. The first results shows that we can efficiently minimise many functions using this approach:
\begin{lem}
	Let $v(x)$ be any function on the lattice $\{0, \ldots, k_i-1\}$. Then $v$ can be written as the difference of two submodular functions.
\end{lem}
\begin{proof}
	Let $g$ be any \textit{strictly} submodular function on the same lattice, that is, one with all second differences over pairs $i \neq j$ strictly less than zero. We shall denote by $m$ the minimum absolute value of the second differences over all pairs $i \neq j$ and $x$, that is:
	\begin{equation}
	m = \min_{i \neq j, x} |g(x + e_i + e_j) - g(x+e_i) - g(x+e_j) + g(x)|.
	\end{equation}
	As all second differences of $g$ are less than zero as its submodular, note for any other difference $m'$ we have $\frac{m'}{m} \leq -1$. We shall additionally define $n$ to be the maximum absolute value of the second difference of $v$ under the same constraint:
	\begin{equation}
	n = \max_{i \neq j, x}  |v(x + e_i + e_j) - v(x+e_i) - v(x+e_j) + v(x)|.
	\end{equation}
	We can now form the function $f$:
	\begin{equation}
		f(x) = v(x) + \frac{n}{m} g(x).
	\end{equation}
	We claim that $f$ is submodular, which will give our result, as a submodular function scaled by a constant is submodular. To do this, take any of the second differences over pairs $i \neq j$ of the function $f$. Denote the second difference as $D_{i,j}(f)(x)$. We find:
	\begin{align}
	D_{i,j}(f)(x) &= D_{i,j}(v)(x) + \frac{n}{m} D_{i,j}(g)(x) \\
	&\leq D_{i,j}(v)(x) - n \\
	&\leq 0,
	\end{align}
	as required.
\end{proof}

Of course, $n,m$ is in general hard to find, so another question that may be asked regarding this problem is how difficult it is to find the functions $f, g$ corresponding to some $v$. That is answered with this result, extending \cite{iyer2012algorithms}:
\begin{lem}
	Suppose that we know $n$, or a lower bound on $n$ as in the previous lemma. Then given a function $v$, we can construct submodular functions $f, g$ such that $v(x) = f(x) - g(x)$.
\end{lem}
\begin{proof}
	Consider the submodular function on $\prod_i \{0,\ldots,k_i-1\}$:
	\begin{equation}
		g(x_1, \ldots, x_p) = x_1^2 + \ldots + x_p^2 - 4\sum_{i\neq j} x_ix_j
	\end{equation}
We can clearly verify here that $g$ is strictly submodular, and we have $m =-4$. Thus if we know a lower bound on $n$ we can form $f$ in the manner of Lemma $4.3$ as required.
\end{proof}
Note that the choice of $g$ in the proof above is not special, it simply needs to be submodular and we need to know its $m$. Thus depending on the particular problem $v(x)$, different choices of $g$ may give nice or meaningful decompositions.

To do complexity analysis, we will be working on an $\epsilon$-approximate version of the algorithm, introduced in \cite{orlin2004approximate}. This means that we will only proceed to step $t+1$ if we have $v(x^{t+1}) \leq (1+\epsilon)v(x^t)$. The reason we do this comes from \cite{iyer2012algorithms}, as we know this problem for set functions is PLS-complete. We now consider the complexity of this procedure. The complexity in minimising the lattice version can be found from a lemma whose statement and proof can be adapted directly from \cite{iyer2012algorithms}:

\begin{lem}
	The $\epsilon$-approximate version of each algorithm has a worst case complexity of $\mathcal{O}(\frac{\log(|M|/|m|)}{\epsilon}T)$, where $T$ is the complexity of the iterative step, $M = f'(0) - g'(k_{max})$ and $m = v(x^1)$.
\end{lem}
We also note that as described earlier, the ModMod procedure has the lowest complexity at each iteration.

\subsection{Constrained Optimisation}
For the SubSup procedure, we are minimising a submodular function with some constraint at every iteration. However, we know that this is hard, and also hard to approximate \cite{svitkina2011submodular} for the set function case, and so will be for ours also. Therefore this algorithm is not suitable for using constraints.

We know that for some simple constraints such as cardinality, knapsack and polymatroid, the submodular maximisation problem has been studied on the integer lattice for monotone functions \cite{soma2018maximizing}. So at least for some subclasses, we can use the SupSub procedure to optimise under constraints.

While to the best of our knowledge no-one has explicitly written algorithms for modular minimisation on the integer lattice, we know it is easy and can be done exactly at least for cardinality constraints, where we can just enumerate all marginal gains and take the lowest in each variable. We will then take the $B$ lowest, where our cardinality constraint says $x(V) \leq B$. 

In the set function case, we can optimise easily and exactly over a variety of other constraints \cite{iyer2012algorithms}. However, here our separability also gives us linearity, something we lose on our lattice as we separate to arbitrary discrete functions. It is worth looking into more of these constrained optimisation problems.



\section{Conclusion}
In \cite{iyer2012algorithms}, the authors studied the problem of minimising the difference between two submodular set functions. We have extended this to the case of general lattice submodular functions, without the DR requirement. Additionally, we note that via discretisation, our method can be applied to continuous functions also.

In performing the majorisation-minimisation technique, we extended an earlier bound from \cite{iyer2012algorithms} for the upper bound, which is valid for DR-submodular functions, and used a decomposition for submodular functions into DR-submodular functions plus a modular function to take advantage of it. For the lower bound, we used the method of computing the subgradient as in \cite{bach2016submodular}. The result of this greedy algorithm gives us our lower bound, and did not require $f$ to be DR-submodular.

After that we formally stated our algorithms, performed some complexity analysis, and analysed other theoretical properties of it. One clear extension on this work would be finding an alternative lower bound that can be used on continuous functions, as our upper bound can already be used in this context.
\newpage
\bibliographystyle{plain}
\bibliography{DifferenceOfSubmodularDraft}

\begin{thebibliography}{10}

\bibitem{alaei2010maximizing}
Saeed Alaei, Ali Makhdoumi, and Azarakhsh Malekian.
\newblock Maximizing sequence-submodular functions and its application to
  online advertising.
\newblock {\em arXiv preprint arXiv:1009.4153}, 2010.

\bibitem{bach2016submodular}
Francis Bach.
\newblock Submodular functions: from discrete to continuous domains.
\newblock {\em Mathematical Programming}, pages 1--41, 2016.

\bibitem{bian2017continuous}
An~Bian, Kfir Levy, Andreas Krause, and Joachim~M Buhmann.
\newblock Continuous dr-submodular maximization: Structure and algorithms.
\newblock In {\em Advances in Neural Information Processing Systems}, pages
  486--496, 2017.

\bibitem{buchbinder2014submodular}
Niv Buchbinder, Moran Feldman, Joseph~Seffi Naor, and Roy Schwartz.
\newblock Submodular maximization with cardinality constraints.
\newblock In {\em Proceedings of the twenty-fifth annual ACM-SIAM symposium on
  Discrete algorithms}, pages 1433--1452. Society for Industrial and Applied
  Mathematics, 2014.

\bibitem{cunningham1983decomposition}
William~H Cunningham.
\newblock Decomposition of submodular functions.
\newblock {\em Combinatorica}, 3(1):53--68, 1983.

\bibitem{ene2016reduction}
Alina Ene and Huy~L Nguyen.
\newblock A reduction for optimizing lattice submodular functions with
  diminishing returns.
\newblock {\em arXiv preprint arXiv:1606.08362}, 2016.

\bibitem{gottschalk2015submodular}
Corinna Gottschalk and Britta Peis.
\newblock Submodular function maximization on the bounded integer lattice.
\newblock In {\em International Workshop on Approximation and Online
  Algorithms}, pages 133--144. Springer, 2015.

\bibitem{hoi2006batch}
Steven~CH Hoi, Rong Jin, Jianke Zhu, and Michael~R Lyu.
\newblock Batch mode active learning and its application to medical image
  classification.
\newblock In {\em Proceedings of the 23rd international conference on Machine
  learning}, pages 417--424. ACM, 2006.

\bibitem{iwata2001combinatorial}
Satoru Iwata, Lisa Fleischer, and Satoru Fujishige.
\newblock A combinatorial strongly polynomial algorithm for minimizing
  submodular functions.
\newblock {\em Journal of the ACM (JACM)}, 48(4):761--777, 2001.

\bibitem{iyer2012algorithms}
Rishabh Iyer and Jeff Bilmes.
\newblock Algorithms for approximate minimization of the difference between
  submodular functions, with applications.
\newblock {\em arXiv preprint arXiv:1207.0560}, 2012.

\bibitem{lee2009non}
Jon Lee, Vahab~S Mirrokni, Viswanath Nagarajan, and Maxim Sviridenko.
\newblock Non-monotone submodular maximization under matroid and knapsack
  constraints.
\newblock In {\em Proceedings of the forty-first annual ACM symposium on Theory
  of computing}, pages 323--332. ACM, 2009.

\bibitem{lin2012learning}
Hui Lin and Jeff~A Bilmes.
\newblock Learning mixtures of submodular shells with application to document
  summarization.
\newblock {\em arXiv preprint arXiv:1210.4871}, 2012.

\bibitem{moriguchi2012discrete}
Satoko Moriguchi and Kazuo Murota.
\newblock On discrete hessian matrix and convex extensibility.
\newblock {\em Journal of the Operations Research Society of Japan},
  55(1):48--62, 2012.

\bibitem{narasimhan2012submodular}
Mukund Narasimhan and Jeff~A Bilmes.
\newblock A submodular-supermodular procedure with applications to
  discriminative structure learning.
\newblock {\em arXiv preprint arXiv:1207.1404}, 2012.

\bibitem{orlin2004approximate}
James~B Orlin, Abraham~P Punnen, and Andreas~S Schulz.
\newblock Approximate local search in combinatorial optimization.
\newblock {\em SIAM Journal on Computing}, 33(5):1201--1214, 2004.

\bibitem{radanovic2015incentive}
Goran Radanovic and Boi Faltings.
\newblock Incentive schemes for participatory sensing.
\newblock In {\em Proceedings of the 2015 International Conference on
  Autonomous Agents and Multiagent Systems}, pages 1081--1089. International
  Foundation for Autonomous Agents and Multiagent Systems, 2015.

\bibitem{soma2018maximizing}
Tasuku Soma and Yuichi Yoshida.
\newblock Maximizing monotone submodular functions over the integer lattice.
\newblock {\em Mathematical Programming}, 172(1-2):539--563, 2018.

\bibitem{svitkina2011submodular}
Zoya Svitkina and Lisa Fleischer.
\newblock Submodular approximation: Sampling-based algorithms and lower bounds.
\newblock {\em SIAM Journal on Computing}, 40(6):1715--1737, 2011.

\end{thebibliography}

\end{document}